\documentclass[11pt,letterpaper]{article}

\usepackage[top=1in, bottom=1in, left=1in, right=1in]{geometry}

\usepackage{amsmath}
\usepackage{bm}
\usepackage{booktabs, subcaption}
\usepackage{array}
\usepackage{color}
\usepackage{bm}
\usepackage{amsmath}
\usepackage{amssymb}
\usepackage{caption}
\usepackage{subcaption}
\usepackage{graphicx}
\usepackage{booktabs}
\usepackage{lipsum}
\usepackage[table]{xcolor}
\usepackage{tikz}
\usepackage{url}
\usepackage{amsthm}
\usepackage{breakcites}

\DeclareMathOperator*{\tree}{tree}
\DeclareMathOperator*{\block}{block}
\DeclareMathOperator*{\regular}{regular}
\DeclareMathOperator*{\argmin}{arg\,min}
\DeclareMathOperator{\sign}{sign}
\newcommand{\calY}{\mathcal{Y}}
\newcommand{\calS}{(\bX, \by)}
\newcommand{\real}{\mathbb{R}}
\newcommand{\mZ}{\mathbb{Z}}
\newcommand{\bA}{\mathbf{A}}
\newcommand{\bB}{\mathbf{B}}
\newcommand{\bC}{\mathbf{C}}
\newcommand{\bD}{\mathbf{D}}

\newcommand{\bbQ}{\mathbb{Q}}
\newcommand{\bX}{\mathbf{X}}
\newcommand{\by}{\mathbf{y}}

\newcommand{\be}{\mathbf{e}}
\newcommand{\ba}{\mathbf{a}}

\newcommand{\bbeta}{\bm{\beta}}
\newcommand{\bmu}{\bm{\mu}}
\newcommand{\bSigma}{\bm{\Sigma}}
\newcommand{\bzero}{\mathbf{0}}
\newcommand{\mI}{\mathbb{I}}
\newcommand{\calF}{\mathcal{F}}
\newcommand{\prob}{\mathbb{P}}

\newcommand{\KL}{\mathbb{K}\mathbb{L}}
\newcommand{\T}{\intercal}

\newtheorem{theorem}{Theorem} 
\newtheorem{corollary}{Corollary} 
\newtheorem{lemma}{Lemma}
 
\newtheorem{proposition}{Proposition}
\newtheorem{definition}{Definition}

\title{Information Theoretic Limits for Standard and One-Bit Compressed Sensing with Graph-Structured Sparsity}

\author{{\bf Adarsh Barik} \\
	Department of Computer Science \\
	Purdue University \\
	\and
	{\bf Jean Honorio} \\
	Department of Computer Science \\
	Purdue University \\  }

\date{}

\begin{document}

\maketitle

\begin{abstract}
	In this paper, we analyze the information theoretic lower bound on the necessary number of samples needed for recovering a sparse signal under different compressed sensing settings. We focus on the weighted graph model, a model-based framework proposed by \cite{hegde2015nearly}, for standard compressed sensing as well as for one-bit compressed sensing. We study both the noisy and noiseless regimes. Our analysis is general in the sense that it applies to any algorithm used to recover the signal. We carefully construct restricted ensembles for different settings and then apply Fano's inequality to establish the lower bound on the necessary number of samples. Furthermore, we show that our bound is tight for one-bit compressed sensing, while for standard compressed sensing, our bound is tight up to a logarithmic factor of the number of non-zero entries in the signal.   
\end{abstract}


\section{Introduction}
\label{sec:introduction}

Sparsity has been a useful tool to tackle high dimensional problems in many fields such as compressed sensing, machine learning and statistics. Several naturally occurring and artificially created signals manifest sparsity in their original or transformed domain. For instance, sparse signals play an important role in applications such as medical imaging, geophysical and astronomical data analysis, computational biology, remote sensing as well as communications. 

In compressed sensing, sparsity of a high dimensional signal allows for the efficient inference of such a signal from a small number of observations. The true high dimensional sparse signal $\bbeta^* \in \real^d$ is not observed directly but its low dimensional linear transformation $\by = \bX \bbeta^* \in \real^n$ is observed along with the design matrix $\bX \in \real^{n \times d}$. The true high dimensional sparse signal $\bbeta^*$ is inferred from observations $(\bX, \by)$. Many signal acquisition settings such as magnetic resonance imaging \cite{lustig2007sparse} use compressed sensing as their underlying model. As a generalization to standard compressed sensing, one can further transform the measurements. One-bit compressed sensing~\cite{boufounos20081,plan2013robust} considers quantizing the measurements to one bit, i.e., $\by = \sign(\bX \bbeta^*)$. This kind of quantization is particularly appealing for hardware implementations.  

The design matrix $\bX$ is a rank deficient matrix. Therefore, in general, measurements $\by$ lose some signal information. However, it is well known that if $\bX$ satisfies the ``Restricted Isometric Property (RIP)'' and the signal $\bbeta^*$ is $s$-sparse (i.e., contains only $s$ non-zero entries) then a good estimation can be done efficiently using $O(s \log \frac{d}{s})$ samples. In practice, a large class of random design matrices satisfy RIP with high probability. Many algorithms such as CoSamp~\cite{needell2010cosamp}, Subspace Pursuit \cite{dai2008subspace} and and Iterative Hard Thresholding~\cite{blumensath2009iterative} use design matrices satisfying RIP which allows to provide high probability performance guarantees. 

The learning problem in compressed sensing is to recover a signal which is a \emph{good} approximation of the true signal. The goodness of approximation can be measured by either a pre-specified distance between the inferred and the true signal, or by the similarity of their support (i.e., the indices of their non-zero entries). The algorithms for compressed sensing try to provide performance guarantees for either one or both of these measures. For instance, \cite{needell2010cosamp}, \cite{dai2008subspace}  and \cite{blumensath2009iterative} provide performance guarantees in terms of distance, while \cite{karbasi2009support} and \cite{li2015sub} provide performance guarantees in terms of support recovery for standard compressed sensing. \cite{gopi2013one} provide guarantees in terms of both distance and support for one-bit compressed sensing.

\cite{baraniuk2010model} initially proposed a model-based sparse recovery framework. Under this framework, \cite{baraniuk2010model} have shown that the sufficient number of samples for correct recovery is logarithmic with respect to the cardinality of the sparsity model. The model of \cite{baraniuk2010model} considered signals with common sparsity structure and small cardinality. Later, \cite{hegde2015nearly} proposed a weighted graph model for graph-structured sparsity and accompanied it with a nearly linear time recovery algorithm. \cite{hegde2015nearly} also analyzed the sufficient number of samples for efficient recovery.

In this paper, we analyze the necessary condition on the sample complexity for exact sparse recovery. While our proof techniques can also be applied to any model-based sparse recovery framework, we apply our method to get the necessary number of samples to perform efficient recovery on a weighted graph model. We provide results for both the noisy and noiseless regimes of compressed sensing. We also extend our results to one-bit compressed sensing. We note that a lower bound on sample complexity was previously  provided in ~\cite{barik2017information} when the observer has access to only the measurements $\by$. Here, we analyze the more relevant setting in which the observer has access to the measurements $\by$ along with the design matrix $\bX$. Table \ref{table:main results cs1} shows a comparison of our information theoretic lower bounds on sample complexity under different settings with the existing upper bounds available in the literature. Note that our bounds for one-bit compressed sensing are tight, while for standard compressed sensing our bounds are tight up to a factor of $\log s$.   

\begin{table}[!h]
	\caption{Sample Complexity Results for Structured Sparsity Models ($d$ is the dimension of the true signal, $s$ is the signal sparsity, i.e., the number of non-zero entries, $g$ is the number of connected components, $\rho(G)$ is the maximum weight degree of graph $G$, $B$ is the weight budget in the weighted graph model, $K$ is the block sparsity, $J$ is the number of entries in a block and $N$ is the total number of blocks in the block structured sparsity model -- detailed explanation of notations are provided in Sections~\ref{sec: results for weighted graph model} and \ref{sec:specific example})}
	\label{table:main results cs1}
	\begin{subtable}{\textwidth}
		\caption{Standard Compressed Sensing}
		\label{table:main results standard cs1}
		\begin{center}
			\small
			\begin{tabular}{lp{2.7cm}p{2.7cm}l}
				\toprule
				Sparsity Structure & Our Lower Bound & Upper Bound & Reference\\
				\midrule
				Weighted Graph Model & $\tilde{\Omega}(s(\log \rho(G) + \log \frac{B}{s}) + g \log \frac{d}{g})$ & $ O(s(\log \rho(G) + \log \frac{B}{s}) + g \log \frac{d}{g})$ & \cite{hegde2015nearly}  \\
				Tree Structured & $\tilde{\Omega}(s)$ & $ O(s)$ & \cite{baraniuk2010model} \\
				Block Structured & $\tilde{\Omega}(KJ + K \log \frac{N}{K})$ & $ O(KJ + K \log \frac{N}{K})$  & \cite{baraniuk2010model} \\
				Regular $s$-sparsity & $\tilde{\Omega}(s \log \frac{d}{s})$ & $O(s \log \frac{d}{s})$ & \cite{rudelson2005geometric}\\
				\bottomrule
			\end{tabular}
		\end{center}
	\end{subtable}
	
	\bigskip
	
	\begin{subtable}{\textwidth}
		\caption{One-bit Compressed Sensing}
		\label{table:main results one-bit cs1}
		\begin{center}
			\small
			\begin{tabular}{lp{2.7cm}p{2.7cm}l}
				\toprule
				Sparsity Structure & Our Lower Bound & Upper Bound & Reference \\
				\midrule
				Weighted Graph Model  & $ \Omega(s(\log \rho(G) + \log \frac{B}{s}) + g \log \frac{d}{g})$ & Not Available & \\
				Tree Structured  & $\Omega(s)$ & Not Available & \\
				Block Structured  & $\Omega(KJ + K \log \frac{N}{K})$ & Not Available & \\
				Regular $s$-sparsity & $\Omega(s \log \frac{d}{s})$ & $O(s \log \frac{d}{s})$ &  \cite{plan2013robust}\\
				\bottomrule
			\end{tabular}
		\end{center}
	\end{subtable}
\end{table}

The paper is organized as follows. We introduce the problem formally in Section~\ref{sec:problem description}. We briefly describe the weighted graph model in Section \ref{sec: results for weighted graph model}. We state our main results in  Section \ref{sec:main result wgm} and extend them to some specific sparsity structures in Section \ref{sec:specific example}. Section~\ref{sec:proof of main result} provides the construction procedure of restricted ensembles and proofs of our main results. Finally, we make our concluding remarks in Section~\ref{sec:conclusion}.  


\section{Problem Description}
\label{sec:problem description}

In this section, we introduce the observation model and later specialize it for specific problems such as standard compressed sensing and one-bit compressed sensing. 

\subsection{Notation}
\label{subsec:notation}

In what follows, we list down the notations which we use throughout the paper. The unobserved true $d$-dimensional signal is denoted by $\bbeta^* \in \real^d$. The inferred signal is represented by $\hat{\bbeta} \in \real^d$. We call a signal $\bbeta \in \real^d, s < d$ an $s$-sparse signal if $\bbeta$ contains only $s$ non-zero entries. The $n$-dimensional observations are denoted by $\by \in \real^n, n \ll d$. We denote the design matrix by $\bX \in \real^{n \times d}$. The $(i,j)$\textsuperscript{th} element of the design matrix is denoted by $\bX_{ij}, \forall\ 1 \leq i \leq n, 1 \leq j \leq d $. The $i$\textsuperscript{th} row of $\bX$ is denoted by $\bX_{i.},  \forall\ 1 \leq i \leq n, $ and the $j$\textsuperscript{th} column of $\bX$ is denoted by $\bX_{.j},  \forall\ 1 \leq j \leq d, $. We assume that the true signal $\bbeta^*$ belongs to a set $\calF$, which is defined more formally later. The number of elements in a set $A$ is denoted by $|A|$. The measurement vector $\by \in \real^n$ is a function $f(\bX \bbeta^* + \be)$ of $\bX, \bbeta^*$ and $\be$ where $\be \in \real^n$ is Gaussian noise with i.i.d. entries, each with mean $0$ and variance $\sigma^2$. The probability of the occurrence of an event $E$ is denoted by $\prob(E)$. The expected value of random variable $A$ is denoted by $\mathbb{E}(A)$. We denote the mutual information between two random variables $A$ and $B$ by $\mI(A; B)$. The Kullback-Leibler (KL) divergence from probability distribution $A$ to probability distribution $B$ (in that order) is denoted by $\KL(A\|B)$. We denote the $\ell_2$-norm of a vector $\ba$ by $\| \ba \|$. We use $\det(\bA)$ to denote the determinant of a square matrix $\bA$. The shorthand notation $[p]$ is used to denote the set $\{1,2,\dots,p\}$. Other notations specific to weighted graph models are defined later in Section~\ref{sec: results for weighted graph model}.

\subsection{Observation Model}
\label{subsec:linear prediction problem}

We define a general observation model. The learning problem is to estimate the unobserved true $s$-sparse signal $\bbeta^*$ from noisy observations. Since $\bbeta^*$ is a high dimensional signal, we do not sample it directly. Rather, we observe a function of its inner product with the rows of a randomized matrix $\bX$. Formally, the $i$\textsuperscript{th} measurement $\by_i$ comes from the below model,
\begin{align*}
\by_i = f(\bX_{i.} \bbeta^* + \be_i)\ ,
\end{align*}
\noindent
where $f:\mathbb{R} \rightarrow \mathbb{R}$ is a fixed function. We observe $n$ such i.i.d. samples and collect them in measurement vector $\by \in \real^n$. We can express this mathematically by,
\begin{align}
\label{eq:linpred all}
\by = f(\bX \bbeta^* + \be)\ ,
\end{align}
\noindent 
where, for clarity, we have overridden $f$ to act on each row of $\bX \bbeta^* + \be$. Our task is to recover an estimate $\hat{\bbeta} \in \mathbb{R}^d$ of $\bbeta^*$ from the observations $(\bX, \by)$. By choosing an appropriate function $f$, we can describe specific instances of compressed sensing.

\subsubsection{Standard Compressed Sensing}
\label{subsubsec:linear regression}

The standard compressed sensing is a special case of equation~\eqref{eq:linpred all} by choosing $f (x) = x$. Then we simply have,
\begin{align}
\label{eq:linreg1}
\by = \bX \bbeta^* + \be \ .
\end{align}
Based on the model given in equation~\eqref{eq:linreg1}, we define our learning problem as follows.
\begin{definition}[Signal Recovery in Standard Compressed Sensing]
	\label{def:standard compressed sensing problem}
	Given that the measurements $\by \in \real^n$ are generated using equation~\eqref{eq:linreg1}, from a design matrix $\bX \in \real^{n \times d}$ and Gaussian noise $\be \in \real^n$, how many observations ($n$) of the form $(\bX, \by)$ are necessary to recover an $s$-sparse signal $\hat\bbeta$ such that,
	\begin{align*}
	\| \hat\bbeta - \bbeta^* \| \leq C \| \be \| \ ,
	\end{align*} 
	for an absolute constant $C > 0$.
\end{definition}
Note that in a noiseless setup, we have that $\be = \bzero$, and thus we essentially want to recover the true signal $\bbeta^*$ exactly. The sample complexity of sparse recovery for the standard compressed sensing has been analyzed in many prior works. In particular, if the design matrix $\bX$ satisfies the Restricted Isometry Property (RIP) then algorithms such as CoSamp~\cite{needell2010cosamp}, Subspace Pursuit (SP)~\cite{dai2008subspace} and Iterative Hard Thresholding (IHT) \cite{blumensath2009iterative} can recover $\bbeta^*$ quite efficiently and in a stable way with a sample complexity of $O(s \log\frac{d}{s})$. Many algorithms use random matrices such as Gaussian (or sub-Gaussian in general) and Bernoulli random matrices because it is known that these matrices satisfy RIP with high probability~\cite{baraniuk2008simple}. 

One can exploit extra information about the sparsity structure to further reduce the sample complexity. \cite{baraniuk2010model} have showed that model-based frameworks which incorporate extra information on the sparsity structure can have sample complexity in the order of $O(\log |\calF|)$ where $\calF$ is number of possible supports in the sparsity model, i.e., the cardinality of the sparsity model. In the same line of work, \cite{hegde2015nearly} proposed a weighted graph based sparsity model which can be used to model many commonly used sparse signals. \cite{hegde2015nearly} also provide a nearly linear time algorithm to efficiently learn $\bbeta^*$.

\subsubsection{One-bit Compressed Sensing}
\label{subsubsec:classification}

The problem of signal recovery in one-bit compressed sensing has been introduced recently \cite{boufounos20081}. In this setup, we do not have access to linear measurements but rather observations come in the form of a single bit. This can be modeled by choosing $f(x) = \sign(x)$ or in other words, we have,
\begin{align}
\label{eq:linreg2}
\by = \sign(\bX \bbeta^* + \be) \ .
\end{align}
\noindent Note that we lose lot of information by limiting the observations to a single bit. It is known that for the noiseless case, unlike standard compressed sensing, one can only recover $\bbeta^*$ up to scaling\cite{plan2013robust}. We define our learning problem in this setting as follows.

\begin{definition}[Signal Recovery in One-bit Compressed Sensing]
	\label{def:one-bit compressed sensing problem}
	Given that the measurements $\by \in \{-1,+1\}^n$ are generated using equation~\eqref{eq:linreg2}, from a design matrix $\bX \in \real^{n \times d}$ and Gaussian noise $\be \in \real^n$, how many observations ($n$) of the form $(\bX, \by)$ are necessary to recover an $s$-sparse signal $\hat\bbeta$ such that,
	\begin{align*}
	\| \frac{\hat\bbeta}{\| \hat\bbeta \|} - \frac{\bbeta^*}{\| \bbeta^* \|} \| \leq \epsilon \ ,
	\end{align*} 
	for some $\epsilon \geq 0$.
\end{definition}
Prior works \cite{plan2013robust,gupta2010sample, ai2014one,gopi2013one} have proposed algorithms and analyzed the sufficient number of samples required for sparse recovery. 

The use of model-based frameworks for one-bit compressed sensing is an open area of research. We provide results assuming that $\bbeta^*$ comes from a weighted graph model. These results naturally extend to the regular $s$-sparse signals with no structures (analyzed in the literature above) as well because the weighted graph model subsumes regular $s$-sparsity. In this way, our approach not only provides results for the current state-of-the-art but also provides impossibility results for algorithms which will possibly be developed in the future for the more sophisticated weighted graph model.      

\subsection{Problem Setting}
\label{subsec:problem setting}

In this paper, we establish a bound on the necessary number of samples needed to infer the sparse signal effectively from a general framework. We assume that the nature picks a true $s$-sparse signal $\bbeta^*$ uniformly at random from a set of signals $\calF$. Then observations are generated using the model described in equation~\eqref{eq:linpred all}. The function $f$ is chosen appropriately for different settings. We also assume that the observer has access to the design matrix $\bX$. Thus, observations are denoted by $(\bX, \by)$. This procedure can be interpreted as a Markov chain which is described below:
\begin{align*}
\bbeta^* \rightarrow \big(\bX, \by = f(\bX \bbeta^* + \be)\big) \rightarrow \hat\bbeta 
\end{align*}

We use the above Markov chain in our proofs. We assume that the true signal $\bbeta^*$ comes from a weighted graph model. We state our results for standard sparse compressed sensing and one-bit compressed sensing. We note that our arguments for establishing information theoretic lower bounds are not algorithm specific. 

A lower bound on the sample complexity for weighted graph models was provided in~\cite{barik2017information}, where the observer does not have access to $\bX$. We analyze the more relevant setting in which the observer has access to $\bX$. Compared to \cite{barik2017information}, we additionally analyze one-bit compressed sensing in detail. We use Fano's inequality~\cite{cover2006elements} to prove our result by carefully constructing restricted ensembles. Any algorithm which infers $\bbeta^*$ from this particular ensemble would require a minimum number of samples. The use of restricted ensembles is customary for information theoretic lower bounds~\cite{santhanam2012information,wang2010information}. 

It is important to mention that results  for efficient recovery in compressed sensing depend on the design matrix satisfying certain properties. We describe this in the next subsection.

\subsection{Restricted Isometry Property (RIP)}
\label{sec:rip}
In compressed sensing, several results (see e.g., \cite{baraniuk2010model,hegde2015nearly}) for efficient recovery require that the design matrix satisfies the Restricted Isometry Property (RIP). We say that a design matrix $\bX \in \mathbb{R}^{n \times d}$ satisfies RIP if there exists a $\delta \in (0, 1)$ such that 
\begin{align*}
(1 - \delta) \|\bbeta\|^2 \leq \|\bX\bbeta\|^2 \leq (1 + \delta) \|\bbeta \|^2, \quad \forall \bbeta \in \mathbb{R}^d
\end{align*}  

Intuitively speaking, one does not want the design matrix to stretch the signal too much in $\ell_2$ norm. Many random matrices satisfy RIP with high probability. In our results, we use Gaussian and Bernoulli design matrices which are proven to satisfy RIP~\cite{baraniuk2008simple}.  For the Gaussian design matrix, the entries of the design matrix are i.i.d. Gaussian with mean $0$ and variance $\frac{1}{n}$. For Bernoulli design matrices the entries are i.i.d. taking values $\frac{1}{\sqrt{n}}$ or $ - \frac{1}{\sqrt{n}}$ with equal probability. 

It is easy to see that for the choices of $\bX$ discussed above, $\|\bX\bbeta\|^2$ concentrates on $\|\bbeta\|^2$ in expectation. That is,
\begin{align*}
\mathbb{E}[\|\bX\bbeta\|^2] = \sum_{i=1}^{n} \sum_{j=1}^{d} \mathbb{E}[(\bX_{ij} \bbeta_j)^2] 
= \frac{1}{n}  \sum_{i=1}^{n} \sum_{j=1}^{d} \bbeta_j^2 
= \|\bbeta\|^2
\end{align*} 

\section{Weighted Graph Model (WGM)}
\label{sec: results for weighted graph model}

We assume that  the true $s$-sparse signal comes from a weighted graph model. This encompasses many commonly seen sparsity patterns in signals such as tree structured sparsity, block structured sparsity as well as the regular $s$-sparsity without any additional structure. Next, we introduce the Weighted Graph Model (WGM) which was proposed by \cite{hegde2015nearly}. We also formally state the sample complexity results from \cite{hegde2015nearly}. 

The Weighted Graph Model is defined on an underlying graph $G = (V,E)$ whose vertices represent the coefficients of the unknown $s-$sparse vector $\bbeta^* \in \real^d$ i.e. $V = [d] = \{1, 2, \dots, d\}$. Moreover, the graph is weighted and thus we introduce a weight function $w : E \rightarrow \mathbb{N}$. Borrowing some notations from \cite{hegde2015nearly},  $w(F)$ denotes the sum of edge weights in a forest $F \subseteq G$, i.e., $w(F) = \sum_{e \in F} w_e$. We also assume an upper bound on the total edge weight which is  called the weight budget and is denoted by $B$. The number of non-zero coefficients of $\bbeta^*$ is denoted by the sparsity parameter $s$. The number of connected components in a forest $F$ is denoted by $g$. The weight-degree $\rho(v)$ of a node $v \in V$ is the largest number of adjacent nodes connected by edges with the same weight, i.e.,
\begin{align*}
\rho(v) = \max_{b \in N} |\{ (v', v) \in E\ | \ w(v',v) = b \}| \ .
\end{align*}  
We define the weight-degree $\rho(G)$ of $G$ to be the maximum weight-degree of any $v \in V$. Next, we define the Weighted Graph Model on coefficients of $\bbeta^*$ as follows:    

\begin{definition}[\cite{hegde2015nearly}]
	\label{def:WGM}
	The $(G, s, g, B)-WGM$ is the set of supports defined as
	\vspace{-1pt}
	\begin{align*}
	\mathbb{M} = \left\lbrace S \subseteq [d] \ | \ |S| = s \ and \ \exists\ F \subseteq G\ with\ V_F = S,  \gamma(F) = g,\ w(F) \leq B \right\rbrace\ ,
	\end{align*}
\end{definition} 
where $\gamma(F)$ is number of connected components in a forest $F$. \cite{hegde2015nearly} provided the following sample complexity result for signal recovery of standard compressed sensing under WGM:
\begin{theorem}[\cite{hegde2015nearly}]
	Let $\bbeta^* \in \real^d$ be in the $(G, s, g, B)-WGM$. Then 
	\begin{align*}
	n = O(s(\log \rho(G) + \log \frac{B}{s}) + g \log \frac{d}{g})
	\end{align*}
	i.i.d. Gaussian observations suffice to estimate $\bbeta^*$. More precisely, let $\be \in \real^n$ be an arbitrary noise vector from equation \eqref{eq:linreg1} and $\bX \in \real^{n \times d}$ be an i.i.d. Gaussian matrix. Then we can efficiently find an estimate $\hat{\bbeta}$ as in Definition~\ref{def:standard compressed sensing problem}, that is, 
	\begin{align*}
	\| \bbeta^* - \hat{\bbeta} \| \leq C \|\be\|\ ,
	\end{align*}
	for an absolute constant $C > 0$.
\end{theorem}

Notice that in the noiseless case, that is, when $\be = \bzero$, they recover the true signal $\bbeta^*$ exactly. We prove that information-theoretically, the bound on the sample complexity of standard compressed sensing in  \cite{hegde2015nearly}  is tight up to a logarithmic factor of sparsity. 

\section{Main results}
\label{sec:main result wgm}

In this section, we state our results for the standard compressed sensing and one-bit compressed sensing. We consider both the noisy and noiseless cases.  We establish an information theoretic lower bound on the sample complexity for signal recovery on a WGM. We state our results more formally in the following subsections.

\subsection{Results for Standard Compressed Sensing}
\label{subsec:results for standard compressed sensing}

For standard compressed sensing, the recovery is not exact for the noisy case but it is sufficiently close to the true signal in $\ell_2$-norm with respect to the noise. Our setup, in this case, uses a Gaussian design matrix. The formal statement of our result is as follows.
\begin{theorem}[Standard Compressed Sensing, Noisy Case]
	\label{thm:noisymainresult1}
	There exists a particular $(G, s, g, B)-WGM$, and a particular set of weights for the entries in the support of $\bbeta^*$ such that nature draws a $\bbeta^* \in \mathbb{R}^d$ uniformly at random and produces a data set $ (\bX, \by = \bX \bbeta^* + \be) \in \real^{n \times d + 1} $ of ${n \in \tilde{o}((s-g) (\log \rho(G) + \log \frac{B}{s-g}) + g \log \frac{d}{g} + (s-g) \log \frac{g}{s-g} + s \log 2)}$ i.i.d. observations as defined in equation~\eqref{eq:linreg1} with $\be_i \stackrel{iid}{\sim} \mathcal{N}(0,\sigma^2), \forall i \in \{1 \dots n \}$ then  ${ \prob(\| \bbeta^* - \hat{\bbeta}  \| \geq C \| \be \|) \geq \frac{1}{10}}$ for $0 < C \leq C_0$ irrespective of the procedure we use to infer $\hat\bbeta$ from $\calS$.
\end{theorem}

We provide a similar result for the noiseless case. In this case recovery is exact. We use a Bernoulli design matrix for our proofs. In what follows, we state our result.
\begin{theorem}[Standard Compressed Sensing, Noiseless Case]
	\label{thm:mainresult1}
	There exists a particular $(G, s, g, B)-WGM$, and a particular set of weights for the entries in the support of $\bbeta^*$ such that if nature draws a $\bbeta^* \in \real^d$ uniformly at random and produces a data set $ (\bX, \by = \bX \bbeta^*) \in \real^{n \times d + 1} $ of ${n \in \tilde{o}((s-g) (\log \rho(G) + \log \frac{B}{s-g}) + g \log \frac{d}{g} + (s - g) \log \frac{g}{s - g} + s \log 2)}$ i.i.d. observations as defined in equation~\eqref{eq:linreg1} with $\be = \bzero$ then $\prob(\bbeta^* \neq \hat{\bbeta}) \geq \frac{1}{2}$ irrespective of the procedure we use to infer $\hat\bbeta$ from $\calS$.  
\end{theorem} 

We note that when $s \gg g$ and $B \geq s-g$ then $\tilde{\Omega}((s-g) (\log \rho(G) + \log \frac{B}{s-g}) + g \log \frac{d}{g} + (s - g) \log \frac{g}{s - g} +  s\log 2) $ is roughly $\tilde{\Omega}(s(\log \rho(G) + \log \frac{B}{s}) + g \log \frac{d}{g})$.

\subsection{Results for One-bit Compressed Sensing}
\label{subsec:results for 1-bit}

In this setting, we provide a construction which works for both the noisy and noiseless case. We use a Gaussian design matrix for both of these setups. Our first result in this setting shows that even in our restricted ensemble recovering the true signal exactly is difficult. 
\begin{theorem}[One-bit Compressed Sensing, Exact Recovery]
	\label{thm:mainresult2}
	There exists a particular $(G, s, g, B)-WGM$, and a particular set of weights for the entries in the support of $\bbeta^*$ such that if nature draws a $\bbeta^* \in \real^d$ uniformly at random and produces a data set $ (\bX, \by = \sign( \bX \bbeta^* + \be)  \in \real^{n \times d + 1} $ of $n \in o((s-g) (\log \rho(G) + \log \frac{B}{s-g}) + g \log \frac{d}{g} + (s - g) \log \frac{g}{s - g} + s \log 2)$ i.i.d. observations as defined in equation~\eqref{eq:linreg2} then $\prob(\bbeta^* \neq \hat{\bbeta}) \geq \frac{1}{2}$ irrespective of the procedure we use to infer $\hat\bbeta$ from $\calS$.  
\end{theorem} 

Our second result provides a bound on the necessary number of samples for approximate signal recovery, which we state formally below. 

\begin{theorem}[One-bit Compressed Sensing, Approximate Recovery]
	\label{thm:noisymainresult2}
	There exists a particular $(G, s, g, B)-WGM$, and a particular set of weights for the entries in the support of $\bbeta^*$ such that if nature draws a $\bbeta^* \in \real^d$ uniformly at random and produces a data set $ (\bX, \by = \sign( \bX \bbeta^* + \be)  \in \real^{n \times d + 1} $ of $n \in o((s-g) (\log \rho(G) + \log \frac{B}{s-g}) + g \log \frac{d}{g} + (s - g) \log \frac{g}{s - g} + s \log 2)$ i.i.d. observations as defined in equation~\eqref{eq:linreg2} then $\prob(\|  \frac{\hat{\bbeta}}{\|\hat{\bbeta}\|}  - \frac{\bbeta^*}{\| \bbeta^* \|} \| \geq \epsilon ) \geq \frac{1}{2}$ for some $\epsilon > 0$ irrespective of the procedure we use to infer $\hat\bbeta$ from $\calS$.  
\end{theorem} 

\section{Specific Examples}
\label{sec:specific example}
Our proof techniques can be applied to prove lower bounds of the sample complexity for several specific sparsity structures as long as one can bound the cardinality of the model. Below, we provide information theoretic lower bounds on sample complexity for some well-known sparsity structures.

\subsection{Tree-structured sparsity model}
\label{subsec:tree sparsity}
The tree-sparsity model~\cite{baraniuk2010model} is used in several applications such as wavelet decomposition of piecewise smooth signals and images. In this model, one assumes that the coefficients of the $s-$sparse signal form a $k$-ary tree and the support of the sparse signal form a rooted and connected sub-tree on $s$ nodes in this $k-$ary tree. The arrangement is such that if a node is part of this subtree then the parent of such node is also included in the subtree. Let $T$ be a rooted binary tree on $[d]$ nodes. For any node $i$, $\pi_T(i)$ denotes the parent of node $i$ in $T$. Then, a tree-structured sparsity model $\mathbb{M}_{\tree}$ on the tree $T$ is the set of supports defined as  
\begin{align}
\label{eq:M tree}
\mathbb{M}_{\tree} = \{ S \subseteq [d] \ | \ |S| = s, \text{ and } \forall i \in S,\ \pi_T(i) \in S \} .
\end{align}  
The following corollary provides information theoretic lower bounds on the sample complexity for $\mathbb{M}_{\tree}$.
\begin{corollary}
	\label{cor:tree sparse}
	There exists a binary tree-structured sparsity model, such that
	\begin{enumerate}
		\item If $n \in \tilde{o}(s)$ then $\mathbb{P}(\| \bbeta^* - \hat{\bbeta}  \|  \geq C \| e \|) \geq \frac{1}{10}$ for noisy standard compressed sensing.
		\item If $n \in \tilde{o}(s)$ then $\mathbb{P}( \bbeta^* \ne \hat\bbeta ) \geq \frac{1}{2}$ for noiseless standard compressed sensing.
		\item If $n \in o(s)$ then  $\prob(\| \frac{\hat\bbeta}{\| \hat\bbeta \|} - \frac{\bbeta^*}{\| \bbeta^* \|}  \| \geq \epsilon ) \geq \frac{1}{2}$ for one-bit compressed sensing.
		\item If $n \in o(s)$ then  $\prob(\hat{\bbeta \ne \bbeta^*}) \geq \frac{1}{2}$ for one-bit compressed sensing.
	\end{enumerate} 
\end{corollary}

\subsection{Block sparsity model}
\label{subsec: block sparse}
In the block sparsity model \cite{baraniuk2010model}, an $s-$sparse signal, $\bbeta \in \mathbb{R}^{J \times N}$, can be represented as a matrix with $J$ rows and $N$ columns. The index $ij$ denotes the index of the entry of $\bbeta$ at $i$\textsuperscript{th} row and $j$\textsuperscript{th} column. The support of $\bbeta$ comes from $K$ columns of this matrix such that $s = J K$. More precisely, a block sparsity model $\mathbb{M}_{\block}$ is a set of supports defined as
\begin{align}
\label{eq:M block}
\mathbb{M}_{\block}= \left\lbrace S \subseteq [J \times N] \ | \ \forall i \in [J], j \in L , ij \in S \text{ and } L \subseteq \{1,\dots, N\},\ |L|=K \right\rbrace \ .
\end{align}

\begin{figure}[!t]
	\centering
	\includegraphics[width=1in]{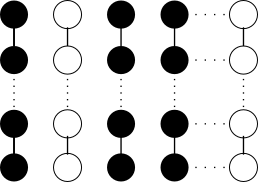}
	\caption{Block sparsity structure as a weighted graph model: nodes are variables, black nodes are selected variables}
	\label{figblock}
\end{figure}
\noindent The above can be modeled as a weighted graph model. In particular, we can construct a graph $G$ over all the entries in $\bbeta$ by treating nodes in the column of the matrix as connected nodes (see Figure~\ref{figblock}). The following corollary provides information theoretic lower bounds on the sample complexity for $\mathbb{M}_{\block}$.

\begin{corollary}
	\label{cor:block sparse}
	There exists a block structured sparsity model, such that
	\begin{enumerate}
		\item If $n \in \tilde{o}(KJ + K \log \frac{N}{K})$ then $\mathbb{P}(\| \bbeta^* - \hat{\bbeta}  \|  \geq C \| e \|) \geq \frac{1}{10}$ for noisy standard compressed sensing.
		\item If $n \in \tilde{o}(KJ + K \log \frac{N}{K})$ then $\mathbb{P}( \bbeta^* \ne \hat\bbeta ) \geq \frac{1}{2}$ for the noiseless standard compressed sensing.
		\item If $n \in o(KJ + K \log \frac{N}{K})$ then  $\prob(\| \frac{\hat\bbeta}{\| \hat\bbeta \|} - \frac{\bbeta^*}{\| \bbeta^* \|}  \| \geq \epsilon ) \geq \frac{1}{2}$ for one-bit compressed sensing.
		\item If $n \in o(KJ + K \log \frac{N}{K})$ then  $\prob( \hat{\bbeta \ne \bbeta^*} ) \geq \frac{1}{2}$ for one-bit compressed sensing.
	\end{enumerate} 
\end{corollary}

\subsection{Regular $s$-sparsity model}
\label{subsec: s sparse}
When the model does not have any additional structure besides sparsity, we call it a regular $s$-sparsity model. That is, a regular $s$-sparsity model $\mathbb{M}_{\regular}$ is a set of supports defined as
\begin{align}
\label{eq:M regular}
\mathbb{M}_{\regular} = \{ S \subseteq [d] | |S| = s \}
\end{align}
The following corollary provides information theoretic lower bounds on the sample complexity for $\mathbb{M}_{\regular}$.
\begin{corollary}
	\label{cor:regular sparse}
	There exists a regular $s$-sparsity model, such that
	\begin{enumerate}
		\item If $n \in \tilde{o}(s \log \frac{d}{s})$ then $\mathbb{P}(\| \bbeta^* - \hat{\bbeta}  \|  \geq C \| e \|) \geq \frac{1}{10}$ for noisy standard compressed sensing.
		\item If $n \in \tilde{o}(s \log \frac{d}{s})$ then $\mathbb{P}( \bbeta^* \ne \hat\bbeta ) \geq \frac{1}{2}$ for noiseless standard compressed sensing.
		\item If $n \in o(s \log \frac{d}{s})$ then  $\prob(\| \frac{\hat\bbeta}{\| \hat\bbeta \|} - \frac{\bbeta^*}{\| \bbeta^* \|}  \| \geq \epsilon ) \geq \frac{1}{2}$ for one-bit compressed sensing.
		\item If $n \in o(s \log \frac{d}{s})$ then  $\prob( \hat{\bbeta \ne \bbeta^*} ) \geq \frac{1}{2}$ for one-bit compressed sensing.
	\end{enumerate} 
\end{corollary}

\section{Proof of Main Results}
\label{sec:proof of main result}

In this section, we prove our main results stated in Section~\ref{sec:main result wgm}. We use the Markov chain described in subsection~\ref{subsec:problem setting} in our proofs. We assume that nature picks a true $s$-sparse signal, $\bbeta^* \in \real^d$, uniformly at random from a family of signals, $\calF$. The definition of $\calF$ varies according to the specific setups. Nature then generates independent and identically distributed samples using the true $\bbeta^*$. These samples are of the form $(\bX, \by = f(\bX \bbeta^* + \be))$. We choose an appropriate $f$ and noise $\be$ for the specific setups under analysis. Similarly, the design matrix $\bX$ also varies according to the specific settings. Although, we note that choice of $\bbeta^*$ and $\bX$ are marginally independent in all the settings. 

The outline of the proof is as follows:
\begin{enumerate}
	\item We define a restricted ensemble $\calF$ and establish a lower bound on the number of possible signals $\bbeta$ in $\calF$.
	\item We obtain an upper bound on the mutual information between the true signal $\bbeta^*$ and the observations $(\bX, \by)$.
	\item We use Fano's inequality~\cite{cover2006elements} to obtain an information theoretic lower bound.  
\end{enumerate} 

We explain each of these steps in detail in the subsequent subsections.

\subsection{Restricted ensemble}
\label{subsec:restrcited ensemble for WGM}

First, we need to construct a weighted graph $G$ to define our family of sparse signals $\mathcal{F}$. Our construction of weighted graph follows \cite{barik2017information}. We construct the underlying graph $G$ for the WGM using the following steps:

\begin{figure*}[h]
	\includegraphics[width=\textwidth]{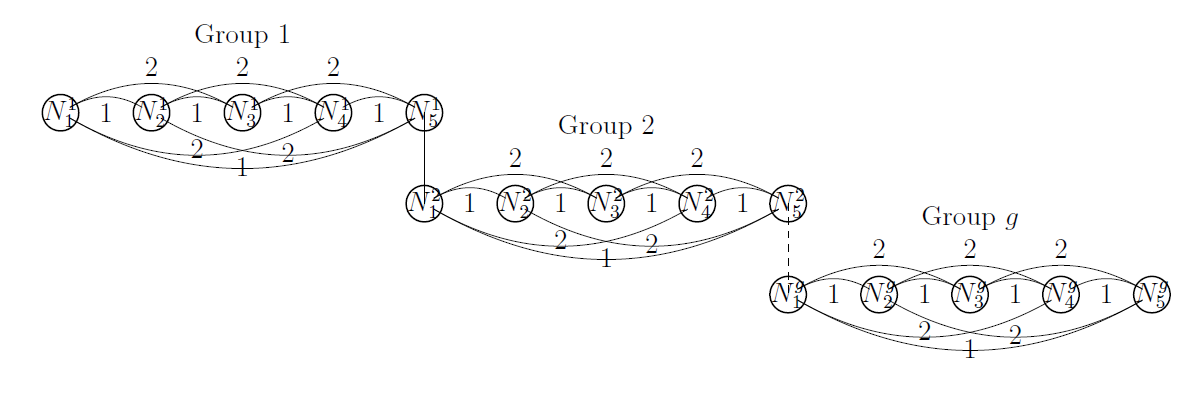}
	\centering \caption{An example of constructing an underlying graph for $\rho(G) = 2$ and $\frac{B}{s-g} = 2$}
	\label{figgraphconst}
\end{figure*}
\begin{figure*}[h]
	\includegraphics[width=\textwidth]{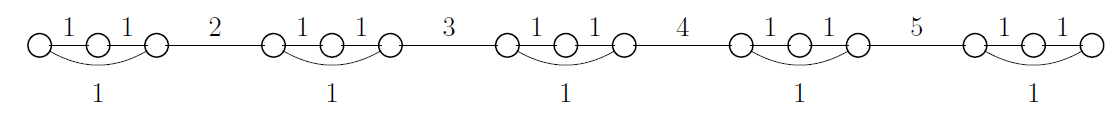}
	\centering\caption{An example of an underlying graph $G$ for $(G, s, g, B)-WGM$ with parameters $d=15, s=10,  g=5, B=5, \rho(G) = 2$}
	\label{figgraphex}
\end{figure*}

\begin{itemize}
	\item We split $d$ nodes equally into $g$ groups with each group having $\frac{d}{g}$ nodes.
	\item For each group $j$, we denote a node by $N_i^j$ where $j$ is the group index and $i$ is the node index. Each group $j$, contains nodes from $N_1^j$ to $N_{\frac{d}{g}}^j$.
	\item We allow for circular indexing within a group, i.e., for any group $j$, a node $N_i^j$ is equivalent to node $N_{i \mod \frac{d}{g}}^j$.  
	\item For each $p = 1,\dots,\frac{B}{s-g}$, node $N_i^j$ has an edge with nodes $N_{i+(p -1)\frac{\rho(G)}{2} + 1}^j$ to $N_{i+p \frac{\rho(G)}{2}}^j$ with weight $p$.
	\item Cross edges between nodes in two different groups are allowed as long as edge weights are greater than $\frac{B}{s-g}$ and this does not affect $\rho(G)$. 
\end{itemize}
Figure~\ref{figgraphconst} shows an example of a graph constructed using the above steps. Furthermore, the parameters of our $WGM$ satisfy the following requirements:
\begin{description}
	\item[R1] $\frac{d}{g} \geq \frac{\rho(G)B}{s-g} + 1$,
	\item[R2] $\frac{\rho(G)B}{2(s-g)} \geq \frac{s}{g} - 1$,
	\item[R3] $B \geq s-g$\ .
\end{description} 
These are quite mild requirements on the parameters and are easy to be fulfilled. To that end, we prove the following.
\begin{proposition}[\cite{barik2017information}]
	\label{prop: mild requirements}
	Given any value of $s, g$ that satisfy R3 (i.e., $B \geq s - g$), there are infinitely many choices for $\rho(G)$ and $d$ that satisfy R1 and R2 and hence, there are infinitely many $(G, s, g, B)$-WGM which follow our construction. 
\end{proposition}

Figure~\ref{figgraphex} shows one graph which follows our construction and additionally fulfills R1, R2 and R3. Now that we have defined the underlying weighted graph $G$ for our WGM, we next define the possible coefficients for the true signal in our $(G,s ,g, B)$-WGM. We define the restricted ensemble for each setting in a different fashion. 

\subsubsection{Restricted Ensemble for Noisy Standard Compressed Sensing}
For the noisy case of standard compressed sensing, our restricted ensemble $\calF_1$ on $G$ is as defined as:
\begin{align}
\label{eq:F noisy wgm}
\begin{split} 
\calF_1 &= \left\lbrace \bbeta\ |\ \bbeta_i = 0, \text{ if } i \notin S,\  \bbeta_i \in \left\lbrace \frac{C_0\sigma \sqrt{n}}{\sqrt{2(1 - \epsilon)}},  \frac{C_0\sigma \sqrt{n}}{\sqrt{2(1 - \epsilon)}} + \frac{C_0\sigma \sqrt{n}}{\sqrt{(1 - \epsilon)}} \right\rbrace,  \text{ if } i \in S,\ S \in \mathbb{M} \right\rbrace\ ,
\end{split} 
\end{align}
for some $0 < \epsilon < 1$ and $\mathbb{M}$ is as in Definition~\ref{def:WGM} in our restricted $(G,s ,g, B)$-WGM. 

\subsubsection{Restricted Ensemble for Noiseless Standard Compressed Sensing}
For the noiseless case, we simplify our ensemble as follows:
\begin{align}
\label{eq:F noiseless wgm}
\begin{split} 
\calF_2 &= \left\lbrace \bbeta\ |\ \bbeta_i = 0, \text{ if } i \notin S,\  \bbeta_i \in \left\lbrace -1,  1 \right\rbrace,  \text{ if } i \in S,\ S \in \mathbb{M} \right\rbrace\ ,
\end{split} 
\end{align}
where $\mathbb{M}$ is as in Definition~\ref{def:WGM} in our restricted $(G,s ,g, B)$-WGM.

\subsubsection{Restricted Ensemble for One-bit Compressed Sensing}
For one-bit compressed sensing, we define our ensemble in the following way:
\begin{align}
\label{eq:F onebit wgm}
\begin{split} 
\calF_3 &= \left\lbrace \bbeta\ |\ \bbeta_i = \begin{cases}
-\epsilon, \text{ if } i \in A\\
\sqrt{\frac{2}{s}} + \epsilon, \text{ if } i \in B\\
0, \text{ otherwise}
\end{cases},  |A| = |B|, S = A \cup B, A \cap B = \phi, S \in \mathbb{M} \right\rbrace\ ,
\end{split} 
\end{align}
for some $\epsilon > 0$ and $\mathbb{M}$ is as in Definition~\ref{def:WGM} in our restricted $(G,s ,g, B)$-WGM.

Next, we count the number of elements in $\calF_1,\calF_2$ and $\calF_3$. We provide the results in the following lemmas.

\begin{lemma}[Cardinality of Restricted Ensemble for Standard Compressed Sensing \cite{barik2017information}]
	\label{lem:cardinality of restricted ensemble}
	For any $\calF \in \{\calF_1, \calF_2\}$ as defined in equations~\eqref{eq:F noisy wgm}, \eqref{eq:F noiseless wgm}, we have that $|\calF| \geq  2^s (\frac{d}{g})^g (\frac{\rho(G)Bg}{2(s-g)^2})^{(s-g)}$.
\end{lemma}

\begin{lemma}[Cardinality of Restricted Ensemble for One-bit Compressed Sensing]
	\label{lem:cardinality of restricted ensemble onebit}
	For $\calF_3$ as defined in equation~\eqref{eq:F onebit wgm}, $|\calF_3| \geq  2^{\frac{s}{2}} (\frac{d}{g})^g (\frac{\rho(G)Bg}{2(s-g)^2})^{(s-g)}$.
\end{lemma}
\begin{proof}
	The proof follows a similar approach as in Lemma \ref{lem:cardinality of restricted ensemble}. Steps 1 and 2 (See \cite{barik2017information}) are the same. For step 3, we choose $\frac{s}{2}$ entries in the support of $\bbeta$ in ${s \choose \frac{s}{2}}$ ways. Thus, 
	\begin{align*}
	\begin{split}
	|\calF_3| &\geq {s \choose \frac{s}{2}} (\frac{d}{g})^g ({\frac{\rho(G)B}{2(s - g)} \choose \frac{s}{g} -1 })^g \\
	&\geq 2^{\frac{s}{2}} (\frac{d}{g})^g (\frac{\rho(G)Bg}{2(s-g)^2})^{(s-g)} \ . 
	\end{split}
	\end{align*}
\end{proof}

%
%

\subsection{Bound on Mutual Information}
\label{subsec:bound on mutual info weighted graph model}

In this subsection, we provide an upper bound on the mutual information between the $s$-sparse signal $\bbeta^* \in \calF$ and the observations $(\bX, \by)$. We provide three lemmas, one for each of the restricted ensembles defined in equations \eqref{eq:F noisy wgm}, \eqref{eq:F noiseless wgm} and \eqref{eq:F onebit wgm}. First, we analyze the noisy case of standard compressed sensing.

\begin{lemma}
	\label{lem:mi bound wgm noisy}
	Let $\bbeta^*$ be chosen uniformly at random from $\calF_1$ defined in equation~\eqref{eq:F noisy wgm} then,
	\begin{align*}
	\mI(\bbeta^*; (\bX, \by)) \leq \frac{n}{2} \log (1 + \frac{s}{2} k_1^2 + \frac{s}{2} k_2^2 - \frac{s^2}{d} \frac{k_1^2}{4} - \frac{s^2}{d} \frac{k_2^2}{4}  - \frac{s^2}{d} \frac{k_1 k_2}{2})  \ ,
	\end{align*}
	for some constants $k_1$ and $k_2$ independent of $n, d$ and $s$ and where $(\bX, \by)$ is generated using the noisy setup defined in equation \eqref{eq:linreg1}.
\end{lemma}
\begin{proof}
	Note that $\by_i = \bX_{i.} \bbeta + \be_i, \forall i \in \{1,\dots,n\}$. Furthermore, in the noisy setup we choose the Gaussian design matrix. That is, each $\bX_{ij}, \forall i \in [n], j \in [d]$ is drawn independently from $\mathcal{N}(0, \frac{1}{n})$. First, we show that for a given $\bbeta$, $(\bX_{i.}, \by_i) \in \real^{d + 1}, \forall i \in \{1,\dots,n\}$ follows multivariate normal distribution $\mathcal{N}(\bmu, \bSigma_{\bbeta})$ where,
	\begin{align}
	\label{eq:mean and sigma joint normal}
	\begin{split}
	&\bmu = \bzero \\
	&\bSigma_{\bbeta} = \begin{bmatrix} \frac{1}{n} & 0 & \dots & \frac{\bbeta_1}{n} \\
	0 & \frac{1}{n} & \dots & \frac{\bbeta_2}{n} \\
	\vdots & \vdots & \vdots & \vdots \\
	\frac{\bbeta_1}{n} & \frac{\bbeta_2}{n} & \dots & \frac{\|\bbeta\|_2^2}{n} + \sigma^2
	\end{bmatrix}
	\end{split}
	\end{align} 
	It can be easily verified that $\ba^\T (\bX_{i.}, \by_i)$ follows normal distribution for any given $\ba \in \real^{d+1}$. This implies that $(\bX_i, \by_i)$ follows a multivariate normal distribution. Second, since each $\bX_{ij} \sim \mathcal{N}(0, \frac{1}{n})$ and $\be_i \sim \mathcal{N}(0, \sigma^2)$, thus $\bmu = \mathbb{E}(\bX_i, \by_i) = \bzero \in \real^{d+1}$. To compute the covariance matrix, we recall that $\bX_{ij}$ are independently distributed. Therefore $\text{Cov}(\bX_{ij}, \bX_{ik}) = 0, \forall j \ne k$. Thus,
	\begin{align*}
	\text{Cov}(\bX_{ij}, \by_i) = \text{Cov}(\bX_{ij}, \sum_{j=1}^d \bbeta_j \bX_{ij} + \be_i) = \frac{\bbeta_j}{n} \\
	\text{Cov}(\by_i, \by_i) = \frac{\|\bbeta\|_2^2}{n} + \sigma^2
	\end{align*}
	Note that for any arbitrary distribution $\bbQ$ over $\calS$, the following inequality holds (See equation 5.1.4 in \cite{duchifano}).	
	\begin{align}
	\label{eq:mutual information inequality}
	\mI(\bbeta^*,\calS) \leq \frac{1}{|\calF_1|} \sum_{\bbeta \in \calF_1} \KL(\prob_{\calS|\bbeta} \| \bbQ)
	\end{align}
	We choose a $\bbQ$ which decomposes in the following way:
	\begin{align*}
	\bbQ = Q^n
	\end{align*}
	Using the independence of samples and factorization of $\bbQ$, we can write equation \eqref{eq:mutual information inequality} as,
	\begin{align}
	\label{eq:mI 2}
	\mI(\bbeta^*,\calS) \leq \frac{n}{|\calF_1|} \sum_{\bbeta \in \calF_1} \KL(\prob_{(\bX_i, \by_i)|\bbeta} \| Q)
	\end{align}
	Recall that $(\bX_i, \by_i) | \bbeta \sim \mathcal{N}(\bzero, \bSigma_{\bbeta}), \forall \bbeta \in \calF_1$ where $\bSigma_{\bbeta}$ is computed according to equation~\eqref{eq:mean and sigma joint normal}. Let $Q \sim \mathcal{N}(\bzero, \bSigma)$. By the KL divergence between two multivariate normal distributions, we can write equation \eqref{eq:mI 2} as,
	\begin{align}
	\label{eq:mI 3}
	\mI(\bbeta^*,\calS) \leq \frac{n}{2} \frac{1}{|\calF_1|} \sum_{\bbeta \in \calF_1}  [\log \frac{\det(\bSigma)}{\det(\bSigma_{\bbeta})} - d -1 + \text{tr}(\bSigma^{-1}\bSigma_{\bbeta}) ]
	\end{align}
	We then choose the covariance matrix $\bSigma$ which minimizes equation \eqref{eq:mI 3}. 
	
	\begin{align*}
	\begin{split}
	\bSigma = \argmin_{\bSigma} \frac{n}{2} \frac{1}{|\calF_1|} \sum_{\bbeta \in \calF_1}  [\log \frac{\det(\bSigma)}{\det(\bSigma_{\bbeta})} - d -1 + \text{tr}(\bSigma^{-1}\bSigma_{\bbeta}) ] 
	\end{split}
	\end{align*} 
	We solve the above equation for a positive definite covariance matrix $\bSigma$. This can be easily done by taking the derivative of the equation and equating it to zero. That is,
	\begin{align*}
	\begin{split}
	&\sum_{\bbeta \in \calF_1}  \bSigma^{-1} - \bSigma^{-1} \bSigma_{\bbeta} \bSigma^{-1}  = \bzero \\
	& \bSigma \big[ \sum_{\bbeta \in \calF_1}  \bSigma^{-1} - \bSigma^{-1} \bSigma_{\bbeta} \bSigma^{-1} \big] \bSigma = \bzero \\
	\end{split}
	\end{align*} 
	and therefore:
	\begin{align}
	\label{eq:min sigma}
	\bSigma = \frac{1}{|\calF_1|} \sum_{\bbeta \in \calF_1} \bSigma_{\bbeta}
	\end{align} 
	Substituting value of $\bSigma$ from equation \eqref{eq:min sigma} to equation \eqref{eq:mI 3}, we get:	
	\begin{align}
	\label{eq:mI 4}
	\mI(\bbeta^*,\calS) \leq \frac{n}{2} [\log \det(\frac{1}{|\calF_1|} \sum_{\bbeta \in \calF_1} \bSigma_{\bbeta} ) - \frac{1}{|\calF_1|} \sum_{\bbeta \in \calF_1}  \log \det(\bSigma_{\bbeta}) ]
	\end{align}
	Next, we compute the determinant of the above covariance matrices. 
	
	\noindent {\bf Computing determinant of covariance matrix $\bSigma_{\bbeta}$.} Note that for a block matrix,
	
	\begin{align}
	\label{eq:block matrix determinant}
	\det\big(\begin{bmatrix}
	\bA & \bB \\
	\bC & \bD 
	\end{bmatrix}\big) = \det(\bA) \det(\bD - \bC \bA^{-1} \bB)
	\end{align}   
	provided that $\bA$ is invertible. Note that $\bSigma_{\bbeta} = \begin{bmatrix}
	\bA & \bB \\
	\bC & \bD 
	\end{bmatrix}$ where,
	
	\begin{align*}
	\begin{split}
	&\bA = \text{diag}(\frac{1}{n}) \in \real^{d\times d} \\
	&\bB = \begin{bmatrix}
	\frac{\bbeta_1}{n} \\
	\vdots \\
	\frac{\bbeta_d}{n}
	\end{bmatrix} \in \real^{d \times 1} \\
	&\bC = \bB^\T \in \real^{1 \times d} \\
	&\bD = \frac{\|\bbeta\|_2^2}{n} + \sigma^2 \in \real^{1\times 1}
	\end{split}
	\end{align*}
	Using equation \eqref{eq:block matrix determinant}, it follows that,
	
	\begin{align*}
	\det(\bSigma_{\bbeta}) &= \frac{1}{n^d} (\frac{\|\bbeta\|_2^2}{n} + \sigma^2 - \frac{\|\bbeta\|_2^2}{n}) \\
	&= \frac{\sigma^2}{n^d}
	\end{align*}
	We can simplify equation \eqref{eq:mI 4}:
	
	\begin{align}
	\label{eq:mI 5}
	\mI(\bbeta^*,\calS) \leq \frac{n}{2} [\log \det(\frac{1}{|\calF_1|} \sum_{\bbeta \in \calF_1} \bSigma_{\bbeta} )  - \log \frac{\sigma^2}{n^d}] 
	\end{align}

	\noindent {\bf Computing determinant of covariance matrix $\bSigma$.} Now note that,
	
	\begin{align*}
	\bSigma = \frac{1}{|\calF_1|} \sum_{\bbeta \in \calF_1} \bSigma_{\bbeta}  = \begin{bmatrix} \frac{1}{n} & 0 & \dots & \frac{\bar\bbeta_1}{n} \\
	0 & \frac{1}{n} & \dots & \frac{\bar\bbeta_2}{n} \\
	\vdots & \vdots & \vdots & \vdots \\
	\frac{\bar\bbeta_1}{n} & \frac{\bar\bbeta_2}{n} & \dots & \frac{\sum_{\bbeta \in \calF_1}\|\bbeta\|_2^2}{|\calF_1|n} + \sigma^2
	\end{bmatrix}
	\end{align*}
	where $\bar\bbeta_i = \frac{\sum_{\bbeta \in \calF_1} \bbeta_i}{|\calF_1|}, \forall i \in [d]$. Using the same approach as equation \eqref{eq:block matrix determinant}, we can compute the determinant of $\bSigma$:
	
	\begin{align}
	\label{eq:determinant of average sigma}
	\det(\bSigma) = \frac{\sigma^2 + \frac{\sum_{\bbeta \in \calF_1}\|\bbeta\|_2^2}{|\calF_1|n} - \frac{\sum_{i=1}^d \bar\bbeta_i^2}{n}}{n^d}
	\end{align}
	Each $\bbeta_i \in \{k_1 \sigma \sqrt{n},\ k_2 \sigma \sqrt{n} \forall i \in [d] \}$ for some constants $k_1$ and $k_2$ with equal probability. Each $\bbeta$ is $s$-sparse and all the $\bbeta_i$ are treated equally. That is overall there should be $s |\calF_1|$ non-zero coefficients, half of them are $k_1 \sigma \sqrt{n}$ and the other half are $\ k_2 \sigma \sqrt{n}$. Using this, equation \eqref{eq:determinant of average sigma} implies:
	
	\begin{align*}
	\det(\bSigma) &= \frac{\sigma^2 + \frac{\frac{s|\calF_1|}{2} k_1^2 \sigma^2 n + \frac{s|\calF_1|}{2} k_2^2 \sigma^2 n }{|\calF_1|n} - \frac{\sum_{i=1}^d (\frac{\frac{s}{d} |\calF_1| \frac{k_1 + k_2}{2}}{|\calF_1|})^2}{n}}{n^d} \\
	&= \frac{\sigma^2}{n^d} (1 + \frac{s}{2} k_1^2 + \frac{s}{2} k_2^2 - \frac{s^2}{d} \frac{k_1^2}{4} - \frac{s^2}{d} \frac{k_2^2}{4}  - \frac{s^2}{d} \frac{k_1 k_2}{2} )
	\end{align*}
	
	\noindent Substituting this in equation \eqref{eq:mI 5}, we get
	
	\begin{align*}
	\mI(\bbeta^*,\calS) \leq \frac{n}{2} \log (1 + \frac{s}{2} k_1^2 + \frac{s}{2} k_2^2 - \frac{s^2}{d} \frac{k_1^2}{4} - \frac{s^2}{d} \frac{k_2^2}{4}  - \frac{s^2}{d} \frac{k_1 k_2}{2}) 
	\end{align*}
\end{proof}

Next, we analyze the noiseless case of standard compressed sensing.
\begin{lemma}
	\label{lem:mi bound wgm noiseless}
	Let $\bbeta^*$ be chosen uniformly at random from $\calF_2$ defined in equation~\eqref{eq:F noiseless wgm} then,
	\begin{align*}
	\mI(\bbeta^*,\calS) \leq 3n \log \frac{es}{27}  \ ,
	\end{align*}
	where $\calS$ is generated using the noiseless setup defined in equation \eqref{eq:linreg1}.
\end{lemma}
\begin{proof}
	We make use of the bound from equation \eqref{eq:mI 2}. For the noiseless case, we assume that the entries of $\bX_i$ follow a Bernoulli distribution with $\bX_{ij} \in \{-\frac{1}{\sqrt{n}},\frac{1}{\sqrt{n}}\}$. Now, since $\bbeta$ is an $s-sparse$ vector with binary non-zero entries $\by_i = \sum_{j=1}^d \bbeta_j \bX_{ij}$ takes values in a finite set which we denote as $\calY$. We can compute the size of $\calY$ in the following way. First, note that if $\bbeta_i \in \{0, a, b\}$ then
	\begin{align*}
	\begin{split}
	\calY &= \big\{ \by_i | \by_i = \alpha a + \beta b - \gamma a - \tau b ; \ \alpha + \beta + \gamma + \tau = s;\  \alpha, \beta, \gamma, \tau \in \mZ_{\geq 0} \big\} \\
	|\calY| &= {s + 3 \choose 3} \leq \frac{e^3s^3}{3^3}
	\end{split}
	\end{align*}  
	Now we assume that $Q_{(\bX_i, \by_i)} = \prob_{\bX_i} Q_{\by_i}$. Furthermore, let $Q_{\by_i}$ be a discrete uniform distribution on $\calY$. Recall that $\bX_i$ and $\bbeta^*$ are marginally independent. Thus, $\prob_{(\bX_i, \by_i, \bbeta)} = \prob_{\bX_i} \prob_{\bbeta} \prob_{\by_i | \bX_i, \bbeta}$. We bound the KL divergence between $\prob_{(\bX_i, \by_i)|\bbeta}$ and $Q$ as follows:
	\begin{align*}
	\KL(\prob_{(\bX_i, \by_i)|\bbeta} \| Q_{(\bX_i, \by_i)}) &= - \sum_{\bX_i, \by_i} \prob_{(\bX_i, \by_i)|\bbeta} \log \frac{Q_{(\bX_i, \by_i)}}{\prob_{(\bX_i, \by_i)|\bbeta}} \\
	&= - \sum_{\bX_i, \by_i} \prob_{\bX_i} \prob_{\by_i | \bX_i, \bbeta} \log \frac{Q_{(\bX_i, \by_i)}}{\prob_{\bX_i} \prob_{\by_i| \bX_i, \bbeta}} \\
	&= - \sum_{\bX_i, \by_i = \bbeta^\T \bX_i} \prob_{\bX_i}  \log \frac{\prob_{\bX_i} Q_{\by_i}}{\prob_{\bX_i}} \\
	&= - \sum_{\bX_i, \by_i = \bbeta^\T \bX_i} \prob_{\bX_i}  \log  Q_{\by_i} \\
	&= \log \frac{e^3s^3}{3^3}
	\end{align*}
	Thus,
	\begin{align*}
	\mI(\bbeta^*,\calS) \leq n \log \frac{e^3s^3}{3^3}
	\end{align*}
\end{proof}

The following lemmas provide an upper bound on the mutual information between the true signal $\bbeta^*$ and observed samples $ (\bX, \by = \sign(\bX \bbeta^* + \be))$ as in one-bit compressed sensing.  

\begin{lemma}
	\label{lem:mi bound one bit noisy}
	Let $\bbeta^*$ be chosen uniformly at random from $\calF_3$ defined in equation~\eqref{eq:F onebit wgm} then,
	\begin{align*}
	\mI(\bbeta^*,\calS) \leq 2n \log 2  \ ,
	\end{align*}
	where $\calS$ is generated using the setup defined in equation \eqref{eq:linreg2} for either the noisy or noiseless case.
\end{lemma}
\begin{proof}
	Again, we make use of the bound from equation~\eqref{eq:mI 2}. We first analyze the noisy case. 
	\paragraph{Noisy Case.}
	In this case, $\by_i \in \calY, \forall i \in [d]$ can take two possible values and hence $\calY = \{+1, -1\}$ where $\calY$ is the set of all possible values of $\by_i$. Thus $|\calY| = 2$. Let $\prob_{\by_i  | \bX_{i.}, \bbeta}(\by_i  = +1 ) = \prob(\be_i > - \bX_{i.}^\T \bbeta) = p$ and $\prob_{\by_i  | \bX_{i.}, \bbeta}(\by_i  = -1) =1 - p$. We choose $Q_{(\bX_i, \by_i)} = \prob_{\bX_i} Q_{\by_i}$, where $Q_{\by_i} = \text{Bernoulli}(\frac{1}{2})$. Similar to the proof of Lemma~\ref{lem:mi bound wgm noiseless}, we can bound the KL divergence between $\prob_{(\bX_i, \by_i)|\bbeta}$ and $Q_{(\bX_i, \by_i)}$ as follows:  
	\begin{align*}
	\KL(\prob_{(\bX_i, \by_i)|\bbeta} \| Q_{(\bX_i, \by_i)}) &= - \sum_{\bX_i, \by_i} \prob_{(\bX_i, \by_i)|\bbeta} \log \frac{Q_{(\bX_i, \by_i)}}{\prob_{(\bX_i, \by_i)|\bbeta}} \\
	&= - \sum_{\bX_i, \by_i} \prob_{\bX_i} \prob_{\by_i | \bX_i, \bbeta} \log \frac{Q_{(\bX_i, \by_i)}}{\prob_{\bX_i}  \prob_{\by_i| \bX_i, \bbeta}} \\
	&= - \sum_{\bX_i, \by_i \in \{+1,-1\}} \prob_{\bX_i} \prob_{\by_i| \bX_i, \bbeta} \log \frac{\prob_{\bX_i} Q_{\by_i}}{\prob_{\bX_i}\prob_{\by_i| \bX_i, \bbeta}} \\
	&= p \log 2p + (1 - p) \log 2(1 - p) \\
	&\leq 2\log 2
	\end{align*}
	\noindent Thus,
	\begin{align*}
	\mI(\bbeta^*,\calS) \leq 2n \log 2
	\end{align*}	
	Next we analyze the noiseless case.
	
	\paragraph{Noiseless Case.}
	Again, the number of values that $\by \in \calY$ can take is two, i.e., $|\calY| = 2$. We choose $Q_{(\bX_i, \by_i)} = \prob_{\bX_i} Q_{\by_i}$, where $Q_{\by_i} = \text{Bernoulli}(\frac{1}{2})$. Then using the same approach as above, we get
	\begin{align*}
	\KL(\prob_{(\bX_i, \by_i)|\bbeta} \| Q_{(\bX_i, \by_i)}) &= - \sum_{\bX_i, \by_i} \prob_{(\bX_i, \by_i)|\bbeta} \log \frac{Q_{(\bX_i, \by_i)}}{\prob_{(\bX_i, \by_i)|\bbeta}} \\
	&= - \sum_{\bX_i, \by_i} \prob_{\bX_i} \prob_{\by_i | \bX_i, \bbeta} \log \frac{Q_{(\bX_i, \by_i)}}{\prob_{\bX_i}  \prob_{\by_i| \bX_i, \bbeta}} \\
	&= - \sum_{\bX_i, \by_i = \sign(\bbeta^\T \bX_i)} \prob_{\bX_i}  \log \frac{Q_{\bX_i} Q_{\by_i}}{\prob_{\bX_i}} \\
	&= - \sum_{\bX_i, \by_i = \sign(\bbeta^\T \bX_i)} \prob_{\bX_i}  \log  Q_{\by_i} \\
	&= 2\log 2
	\end{align*}
	\noindent Thus,
	\begin{align*}
	\mI(\bbeta^*,\calS) \leq 2n \log 2
	\end{align*}
\end{proof}

\subsection{Bound on the Inference Error}
\label{subsec:bound on the inference error}

In this subsection, we analyze the inference error using the results from the previous sections and Fano's inequality. If nature chooses $\bbeta^*$ from a restricted ensemble $\calF$ uniformly at random then for the Markov chain described in Section~\ref{sec:problem description}, Fano's inequality \cite{cover2006elements} can be written as
\begin{align}
\label{eq:fano}
\prob(\hat\bbeta \neq \bbeta^*) &\geq 1 - \frac{\mI(\bbeta^*, \calS) + \log 2}{\log |\calF|} \ .
\end{align} 
Since, we have already established bounds on $\mI(\bbeta^*, \calS)$ and $\log |\calF|$, equation \eqref{eq:fano} readily provides a bound on the error of exact recovery. We can use this directly to prove some of our theorems.
\subsubsection{Proof of Theorem~\ref{thm:mainresult1} - Noiseless Case in Standard Compressed Sensing}
Using the results from Lemmas~\ref{lem:mi bound wgm noiseless} and subsection~\ref{subsec:restrcited ensemble for WGM} along with equation~\eqref{eq:fano}, we get:
\begin{align*}
\prob(\hat\bbeta \neq \bbeta^*) &\geq 1 - \frac{ 3n \log \frac{es}{27} + \log 2}{\log \big( 2^s (\frac{d}{g})^g (\frac{\rho(G)Bg}{2(s-g)^2})^{(s-g)}  \big)} 
\end{align*} 
It follows that $\prob(\hat\bbeta \neq \bbeta^*) \geq \frac{1}{2}$ as long as $n \leq  \frac{1}{6}\frac{(s-g) (\log \frac{\rho(G)}{2} + \log \frac{B}{s-g}) + g \log \frac{d}{g} + (s - g) \log \frac{g}{s - g} + s \log 2}{ \log \frac{es}{27} } - \frac{\log 2}{3 \log \frac{es}{27}}$. This proves Theorem~\ref{thm:mainresult1}.

\subsubsection{Proof of Theorem~\ref{thm:mainresult2} - Exact Recovery in One-bit Compressed Sensing}
Using the results from Lemmas \ref{lem:cardinality of restricted ensemble onebit} and \ref{lem:mi bound one bit noisy}  along with equation~\eqref{eq:fano}, we get:
\begin{align*}
\prob(\hat\bbeta \neq \bbeta^*) &\geq 1 - \frac{ 2n \log 2 + \log 2}{\log \big( 2^{\frac{s}{2}} (\frac{d}{g})^g (\frac{\rho(G)Bg}{2(s-g)^2})^{(s-g)}  \big)} 
\end{align*} 
It follows that $\prob(\hat\bbeta \neq \bbeta^*) \geq \frac{1}{2}$ as long as $n \leq  \frac{1}{2}\frac{(s-g) (\log \frac{\rho(G)}{2} + \log \frac{B}{s-g}) + g \log \frac{d}{g} + (s - g) \log \frac{g}{s - g} + \frac{s}{2} \log 2}{ 2 \log 2 } - \frac{1}{2}$. This proves Theorem~\ref{thm:mainresult2}.

Now we prove the remaining results, which require additional arguments. 

\subsubsection{Proof of Theorem~\ref{thm:noisymainresult1} - Noisy Case in Standard Compressed Sensing}
For noisy setups as described in equation \eqref{eq:linreg1}, we are interested in the error bound in terms of the $\ell_2$-norm. We obtain this bound by using the following rule:
\begin{align}
\label{eq:bound on error}
\begin{split}
\prob(\| \hat{\bbeta} - \bbeta^*  \| \geq C\|\be\|) &\geq \prob(\| \bbeta^*  - \hat{\bbeta} \| \geq C\|\be\|, \hat{\bbeta} \neq  \bbeta^* ) \\
&= \prob(\|  \hat{\bbeta} - \bbeta^*  \| \geq C\|\be\| \ | \  \hat{\bbeta} \neq  \bbeta^* ) \prob(\hat\bbeta \neq \bbeta^*)
\end{split}
\end{align}  
We bound the terms $\prob(\|  \hat{\bbeta} - \bbeta^*  \| \geq C\|\be\| \ | \  \hat{\bbeta} \neq  \bbeta^* )$ and $ \prob(\hat\bbeta \neq \bbeta^*)$ separately. 

\begin{lemma}[\cite{barik2017information}]
	\label{lemma: bound on error}
	If $\bbeta^*$ and $\hat{\bbeta}$ come from a family of signals $\calF_1$ as defined in equation~\eqref{eq:F noisy wgm} then
	\begin{enumerate}
		\item For some $C_0 \geq C > 0$
		\begin{align}
		\label{eq:ballcovering}
		\| \bbeta^* - \hat{\bbeta} \| \leq \frac{C_0\sigma \sqrt{n}}{\sqrt{(1-\epsilon)}} \iff \bbeta^* = \hat{\bbeta}\ .
		\end{align}
		\item For some $0 < \epsilon < 1$,
		\begin{align}
		\label{eq:ebound}
		\prob\Big(\|\be \|^2 \leq \sigma^2 \frac{n}{1 - \epsilon}\Big) \geq 1 - \exp\Big(-\frac{\epsilon^2 n}{4}\Big)\ .
		\end{align}
		\item If the above two claims hold then,
		\begin{align}
		\label{eq: concen bound}
		\mathbb{P}\Big(\|  \hat{\bbeta} - \bbeta^*  \| \geq C \| \be \| \ | \ \hat{\bbeta} \neq \bbeta^*  \Big) \geq 1 - \exp\Big(-\frac{\epsilon^2 n}{4}\Big)\ .
		\end{align}
	\end{enumerate}	
\end{lemma} 

Using the results from Lemmas~\ref{lem:mi bound wgm noisy} and \ref{lem:cardinality of restricted ensemble} along with equation~\eqref{eq:fano}, we get:
\begin{align*}
\prob(\hat\bbeta \neq \bbeta^*) &\geq 1 - \frac{\frac{n}{2} \log (1 + \frac{s}{2} k_1^2 + \frac{s}{2} k_2^2 - \frac{s^2}{d} \frac{k_1^2}{4} - \frac{s^2}{d} \frac{k_2^2}{4}  - \frac{s^2}{d} \frac{k_1 k_2}{2})  + \log 2}{\log \big( 2^s (\frac{d}{g})^g (\frac{\rho(G)Bg}{2(s-g)^2})^{(s-g)}  \big)} 
\end{align*} 
It follows that $\prob(\hat\bbeta \neq \bbeta^*) \geq \frac{1}{2}$ as long as $n \leq  \frac{(s-g) (\log \frac{\rho(G)}{2} + \log \frac{B}{s-g}) + g \log \frac{d}{g} + (s - g) \log \frac{g}{s - g} + s \log 2}{  \log (1 + \frac{s}{2} k_1^2 + \frac{s}{2} k_2^2 - \frac{s^2}{d} \frac{k_1^2}{4} - \frac{s^2}{d} \frac{k_2^2}{4}  - \frac{s^2}{d} \frac{k_1 k_2}{2})  } - \frac{\log 2}{\log (1 + \frac{s}{2} k_1^2 + \frac{s}{2} k_2^2 - \frac{s^2}{d} \frac{k_1^2}{4} - \frac{s^2}{d} \frac{k_2^2}{4}  - \frac{s^2}{d} \frac{k_1 k_2}{2}) }$. Now let $n \in \tilde{o}((s-g) (\log \rho(G) + \log \frac{B}{s-g}) + g \log \frac{d}{g} + (s - g) \log \frac{g}{s - g} + s \log 2)$ so that $\prob(\hat\bbeta \neq \bbeta^*) \geq \frac{1}{2}$. Using result from Lemma~\ref{lemma: bound on error} and equation~\eqref{eq:bound on error} we can write
\begin{align*}
\prob\Big(\|  \hat{\bbeta} - \bbeta^*  \| \geq C\|\be\|\Big) 
\geq \Big(1 - \exp\Big(-\frac{\epsilon^2 n}{4}\Big)\Big) \frac{1}{2} \ . 
\end{align*}
We know that $n \geq 1$ and if we choose $\epsilon \geq \sqrt{-4\log 0.8} \sim 0.9448$, then we can write inequality~\eqref{eq:bound on error} as,
\begin{align*}
\prob\Big(\|  \hat{\bbeta} - \bbeta^*  \| \geq C\|\be\|\Big)  &\geq \frac{1}{10} \ .
\end{align*}
This completes the proof of Theorem~\ref{thm:noisymainresult1}.

\subsubsection{Proof of Theorem~\ref{thm:noisymainresult2} - Approximate Recovery in One-bit Compressed Sensing}

First we note that,
\begin{align*}
\| \bbeta \| = 1 \ , \forall  \bbeta \in \calF_3, 
\end{align*}
\noindent and consequently,
\begin{align*}
\| \frac{\hat\bbeta}{\| \hat\bbeta \|} - \frac{\bbeta^*}{\| \bbeta^* \|}  \| &= \| \hat\bbeta - \bbeta^* \| \ .
\end{align*}
Therefore,
\begin{align}
\label{eq:approx error one bit}
\begin{split}
\prob(\| \frac{\hat\bbeta}{\| \hat\bbeta \|} - \frac{\bbeta^*}{\| \bbeta^* \|}  \| \geq \epsilon ) &= \prob(\| \hat\bbeta - \bbeta^* \| \geq \epsilon) \\
&\geq \prob(\| \hat\bbeta - \bbeta^* \| \geq \epsilon, \hat\bbeta \ne \bbeta^*) \\
&= \prob(\| \hat\bbeta - \bbeta^* \| \geq \epsilon\ | \ \hat\bbeta \ne \bbeta^*) \prob( \hat\bbeta \ne \bbeta^*)
\end{split}
\end{align}
Next we prove the following lemma.
\begin{lemma}
	\label{lem:ball covering onebit}
	If $\hat\bbeta, \bbeta^* \in \calF_3$ then $\prob(\| \hat\bbeta - \bbeta^* \| \geq \epsilon\ | \ \hat\bbeta \ne \bbeta^*) = 1$.
\end{lemma}
\begin{proof}
	We prove that two arbitrarily chosen $\bbeta^1$ and $\bbeta^2$ such that $\bbeta^1, \bbeta^2 \in \calF_3$  as defined in equation~\eqref{eq:F onebit wgm} and $\bbeta^1 \neq \bbeta^2$ then $\| \bbeta^1 - \bbeta^2 \| \geq \epsilon$. 
	
	\paragraph{$\bbeta^1$ and $\bbeta^2$ have the same support.} Since we assume that $\bbeta^1 \neq \bbeta^2$, both vectors must differ in at least one coefficient on their support. Let $i$ be such a coefficient. Then,   
	\begin{align*}
	\|\bbeta^1 - \bbeta^2\| &\geq |\bbeta^1_i - \bbeta^2_i| \\
	&=|\sqrt{\frac{2}{s}} + 2 \epsilon| \\
	&\geq \epsilon
	\end{align*}
	\paragraph{$\bbeta^1$ and $\bbeta^2$ have different supports} When $\bbeta^1$ and $\bbeta^2$ have different supports then we can always find $i$ and $j$ such that $i \in S_1, i \notin S_2$ and $j \notin S_1, j \in S_2$ where $S_1$ and $S_2$ are supports of $\bbeta^1$ and $\bbeta^2$ respectively. Then, 
	\begin{align*}
	\| \bbeta^1 - \bbeta^2 \| &\geq \sqrt{(\bbeta^1_i)^2 + (\beta^2_j)^2 } \\
	&\geq \sqrt{\epsilon^2 + \epsilon^2}\\
	&\geq \epsilon \ .
	\end{align*}
	Since the above is true for any two arbitrarily chosen $\bbeta^1$ and $\bbeta^2$, this holds for $\bbeta^*$ and $\hat{\bbeta}$ as well. This proves the lemma.
\end{proof}

Substituting results from Lemma \ref{lem:ball covering onebit} into equation \eqref{eq:approx error one bit}, we get,
\begin{align*}
\begin{split}
\prob(\| \frac{\hat\bbeta}{\| \hat\bbeta \|} - \frac{\bbeta^*}{\| \bbeta^* \|}  \| \geq \epsilon ) &=  \prob( \hat\bbeta \ne \bbeta^*)
\end{split}
\end{align*}
Using the results from Lemmas~\ref{lem:mi bound one bit noisy} and \ref{lem:cardinality of restricted ensemble onebit} along with equation~\eqref{eq:fano}, we get:
\begin{align*}
\prob(\| \frac{\hat\bbeta}{\| \hat\bbeta \|} - \frac{\bbeta^*}{\| \bbeta^* \|}  \| \geq \epsilon ) &\geq 1 - \frac{ 2n \log 2 + \log 2}{\log \big( 2^{\frac{s}{2}} (\frac{d}{g})^g (\frac{\rho(G)Bg}{2(s-g)^2})^{(s-g)}  \big)}  
\end{align*} 
It follows that as long as $n \leq  \frac{1}{2}\frac{(s-g) (\log \frac{\rho(G)}{2} + \log \frac{B}{s-g}) + g \log \frac{d}{g} + (s - g) \log \frac{g}{s - g} + \frac{s}{2} \log 2}{ 2 \log 2 } - \frac{1}{2}$, we get $\prob(\| \frac{\hat\bbeta}{\| \hat\bbeta \|} - \frac{\bbeta^*}{\| \bbeta^* \|}  \| \geq \epsilon ) \geq \frac{1}{2}$ . This completes the proof of Theorem~\ref{thm:noisymainresult2}.

\subsection{Sample Complexity for Specific Sparsity Structures}
\label{subsec:sample complexity for specific sparsity structures}

Until now, the restricted ensembles used in our proofs are defined on a general weighted graph model. These proofs can be instantiated directly to commonly used sparsity structures such as tree structured sparsity, block sparsity and regular $s$-sparsity by defining $\calF_1, \calF_2$ and $\calF_3$ defined in equations \eqref{eq:F noisy wgm}, \eqref{eq:F noiseless wgm}, \eqref{eq:F onebit wgm}, on the models $\mathbb{M}_{\tree}, \mathbb{M}_{\block}$ and $\mathbb{M}_{\regular}$ defined in equations \eqref{eq:M tree}, \eqref{eq:M block} and \eqref{eq:M regular}. We only need to bound the cardinality of these restricted ensembles to extend our proofs to specific sparsity structures.

\paragraph{Proof of Corollary \ref{cor:tree sparse}.}  
We take a restricted ensemble $\calF \in \{\calF_1, \calF_2, \calF_3\}$ which is defined on $\mathbb{M}_{\tree}$. For such a restricted ensemble, we have that $\log |\mathcal{F}| \geq cs$ for an absolute constant $c > 0$. This follows from the fact that we have at least $2^s$ different choices of $\bbeta^*$ for standard compressed sensing and $2^{\frac{s}{2}}$ for one-bit compressed sensing. Using the below results from \cite{baraniuk2010model}, we note that this is not a weak bound as the upper bound is of the same order, i.e.,
\begin{align*}
|\calF| \leq \begin{cases}
2^{s} \frac{2^{2s + 8}}{s e^2}, s \geq \log d \\
2^{s} \frac{(2e)^s}{s+1}, s < \log d
\end{cases}, \forall \calF \in \{\calF_1, \calF_2 , \calF_3 \}
\end{align*}
From the above and using Theorems~\ref{thm:noisymainresult1}, \ref{thm:mainresult1}, \ref{thm:mainresult2},  and \ref{thm:noisymainresult2}, we prove our claim.

\paragraph{Proof of Corollary \ref{cor:block sparse}.}  
We take a restricted ensemble $\calF \in \{\calF_1, \calF_2, \calF_3\}$ which is defined on $\mathbb{M}_{\block}$. For such a restricted ensemble, one can bound $|\mathcal{F}|$ by choosing $K$ connected components from $N$. It is easy to see that the number of possible signals in this model $\calF$, would be, $|\mathcal{F}| \geq 2^{\frac{KJ}{2}} {N \choose K} \geq 2^{\frac{KJ}{2}} (\frac{N}{K})^K$. Given this and Theorems~\ref{thm:noisymainresult1}, \ref{thm:mainresult1}, \ref{thm:mainresult2},  and \ref{thm:noisymainresult2}, we prove our claim.

\paragraph{Proof of Corollary \ref{cor:regular sparse}.}  
In this case, we take a restricted ensemble $\calF \in \{\calF_1, \calF_2, \calF_3\}$ which is defined on $\mathbb{M}_{\regular}$. We can bound the cardinality of such a restricted ensemble $\calF$ in the following way.
\begin{align*}
\begin{split}
|\calF| &\geq 2^{\frac{s}{2}} {d \choose s}, \quad \forall \calF \in \{\calF_1, \calF_2, \calF_3 \} \\
&\geq 2^{\frac{s}{2}} \frac{d^s}{s^s}
\end{split} 
\end{align*} 

\noindent Given the above and Theorems~\ref{thm:noisymainresult1}, \ref{thm:mainresult1}, \ref{thm:mainresult2},  and \ref{thm:noisymainresult2},  we prove our claim.

\section{Concluding Remarks}
\label{sec:conclusion}

In this paper, we provide information theoretic lower bounds on the necessary number of samples required to recover a signal in compressed sensing. For our proofs, we have assumed that the signal has a sparsity structure which can be modeled by weighted graph model. We have provided specific bounds on the sample complexity for many commonly seen sparsity structures including tree-structured sparsity, block sparsity and regular $s$-sparsity. In case of regular $s$-sparsity, the bound on the sample complexity for one-bit compressed sensing is tight as it matches the current upper bound~\cite{rudelson2005geometric}. For standard compressed sensing, our bounds are tight up to a factor of $\frac{1}{\log s}$. The use of the model-based framework for one-bit compressed sensing remains an open area of research and our information theoretic lower bounds on sample complexity may act as a baseline comparison for the algorithms proposed in future.   





\end{document}